%% file: RIMS-SpinBoson.tex
\documentclass[a4paper,12pt,oneside]{article}
\input{header}

\def\titlename{\scshape  Ground States in the Infrared-Critical Spin Boson Model}
\title{\LARGE\scshape   Existence of Ground States in the Infrared-Critical Spin Boson Model}
\newcommand{\shortauthors}{B. Hinrichs}

\author{Benjamin Hinrichs\thanks{Friedrich Schiller University Jena, Department of Mathematics, Ernst-Abbe-Platz 2, 07743 Jena, Germany\\E-Mail: \texttt{benjamin.hinrichs@uni-jena.de}}}
\newcommand{\E}{{\mathds{E}}}

\newcommand{\Ht}{\widetilde{H}}

\newcommand{\Od}{{\Omega_\downarrow}}

\newcommand{\chr}[1]{\mathbf{1}_{#1}}

\newcommand{\pln}{\psi_n^{(\lambda,\mu)}}

\def\endstepsymbol{}

\begin{document}
	
	\maketitle\thispagestyle{empty}\vspace*{-2em}
	\begin{abstract}\noindent
		We review recent results on the existence of ground states for the infrared-critical spin boson model, which describes the interaction of a massless bosonic field with a two-state quantum system. Explicitly, we derive a critical coupling $\lambda_\sfc>0$ such that the spin boson model exhibits a ground state for coupling constants $\lambda$ with $|\lambda|<\lambda_\sfc$. The proof combines a Feynman--Kac--Nelson formula for the spin boson model with external magnetic field, a 1D-Ising model correlation bound and a compactness argument in Fock space. Elaborating on the connection to a long-range 1D-Ising model, we briefly discuss the conjecture that the spin boson model does not have a ground state at large coupling. This note is based on joint work with David Hasler and Oliver Siebert.
	\end{abstract}

\section{Introduction}

In models describing the interaction of a quantum mechanical particle with a quantum field of massless bosons, one encounters an infrared problem. Intuitively, this can be explained by the circumstance that a finite energy fluctuation might lead to the creation of infinitely many low-energy (sometimes called `soft') bosons. Mathematically speaking, such an infrared catastrophe is reflected in the fact that the Hamilton operator describing the system does not exhibit a ground state, i.e., a stable state at lowest energy. However, it has been observed that underlying symmetries of quantum systems can lead to a mutual cancellation of infrared divergences. Prominently, this has been observed for the model of non-relativistic quantum electrodynamics, see for example \cite{GriesemerLiebLoss.2001,BachChenFrohlichSigal.2007,HaslerSiebert.2020} and references therein.

The spin boson model is also a model which exhibits an underlying symmetry. It describes a two-state quantum mechanical system linearly coupled to a field of bosons. Previously, in \cite{HaslerHerbst.2010,BachBallesterosKoenenbergMenrath.2017}, it has been shown that the spin boson model exhibits a ground state for small coupling constants, by perturbative methods. In the articles \cite{HaslerHinrichsSiebert.2021a,HaslerHinrichsSiebert.2021b,HaslerHinrichsSiebert.2021c}, we presented the first non-perturbative proof for existence of ground states below an explicitly derived critical coupling constant. In this note, these articles are reviewed. Additionally, we enhance the compactness argument from \cite{HaslerHinrichsSiebert.2021a} based on arguments presented in \cite{Matte.2016,HiroshimaMatte.2019}, which allows for a more general choice of dispersion relations.

To emphasize the role of the so-called spin flip symmetry in the cancellation of infrared divergences, we discuss the spin boson model in presence of an external magnetic field. Whereas the external magnetic field has no influence on the existence of ground states in the infrared-regular case, it is crucial that the magnetic field vanishes for a ground state to exist if the model is infrared-critical.

\subsection*{Acknowledgements}

I want to thank David Hasler and Oliver Siebert for the fruitful collaboration, which this work is based on. Further, I am grateful to Fumio Hiroshima for inviting me as a speaker at the RIMS conference. I am looking forward to a (hopefully soon) face-to-face version of this meeting. I also want to thank Oliver Matte for pointing me towards the Fr\'echet--Kolmogorov theorem providing a significant simplification of the compactness argument in \cite{HaslerHinrichsSiebert.2021a}.

\section{Notation, Model and Main Result}
Let us start a more thorough discussion of our results with the following simple definition.
\begin{defn}
	Given a selfadjoint lower-semibounded  Hilbert space operator $H$, we say $H$ has a {\em ground state} if $\inf \sigma(H)$ is an eigenvalue. We say $H$ has a {\em unique} ground state if the corresponding eigenspace is one-dimensional.
\end{defn}
Throughout this note, we fix the following three objects describing the spin boson model.
\begin{description}
	\item{\em Dimension:} $d\in\IN$
	\item{\em Dispersion Relation:} selfadjoint multiplication operator $\omega:\IR^d\to[0,\infty)$ acting on $L^2(\IR^d)$ and satisfying $\omega>0$ almost everywhere
	\item{\em Form Factor:} $v\in \cD(\omega^{-1/2})$, where $\cD(\cdot)$ denotes the domain of a Hilbert space operator
\end{description}
We call the spin boson model {\em infrared-regular} if $v\in\cD(\omega^{-1})$. Otherwise, we call it {\em infrared-critical}.

Given $M\subset \IR^d$, we will also use the following notation:
\begin{description}
	\item{$\chr M$ } denotes the characteristic function of $M$.
	\item{$\cC^\infty_M$ } denotes the set of all smooth functions $\varrho:\IR^d\to[0,1]$ satisfying $\supp \varrho \subset M^\sfc$.
\end{description}
Finally, given $x\in\IR^n$ for some $n\in\IN$, we define the translation operator $\tau_x$ acting on $f\in L^2 (\IR^n)$ as $\tau_x f(p)=f(p+x)$.

The definition of the spin boson model requires the usual bosonic Fock space $\FS$ over $L^2(\IR^d)$, given by
\[ \FS = \bigoplus_{n=0}^\infty\FS^{(n)} \qquad \mbox{with}\ \FS^{(0)}=\IC,\ \FS^{(k)}=L^2_{\sfs\sfy\sfm}(\IR^{d\cdot k})\ \mbox{for}\ k\in\IN,  \]
where we symmetrize over the $k$ $d$-dimensional variables.
Further, for a selfadjoint multiplication operator $m:\IR^d\to\IR$, we define its second quantization as
\[ \dG(\omega) = \bigoplus_{n=0}^\infty \dG^{(n)}(\omega), \qquad\mbox{with}\ \dG^{(0)}(\omega)=0\ \mbox{and}\ \dG^{(n)}(\omega)(k_1,\ldots,k_n) = {\sum_{\ell=1}^{n}\omega(k_\ell)}\ \mbox{for}\ n\in\IN.  \]
Given $f\in L^2(\IR^d)$, we define the associated annihilation operator as
\[ (a(f)\psi)^{(n)}(k_1,\ldots,k_n) = \sqrt{n+1}\int \bar{f(k)}\psi^{(n+1)}(k,k_1,\ldots,k_n)\d k \]
and the corresponding field operator as the selfadjoint operator given by
\[ \ph(f)  = \bar{a(f)+a(f)^*},\]
where $\bar{\phantom{(} \cdot\phantom{)} }$ as usually denotes the operator closure.

The spin boson Hamiltonian now acts on the tensor product Hilbert space
\begin{equation}\label{def:HS}
	\HS = \IC^2 \otimes \FS \cong \FS\oplus \FS,
\end{equation}
where the unitary equivalence is implemented by $(\alpha_1,\alpha_2)\otimes \psi \mapsto \alpha_1\psi \oplus \alpha_2 \psi$.
The Hamilton operator itself is defined as
\begin{equation}\label{def:SB}
	H(\lambda,\mu) = \sigma_z\otimes \Id + \dG(\omega) \otimes \Id + \sigma_x \otimes (\lambda \ph(v) + \mu \Id),
\end{equation}
where $\sigma_x=\begin{psmallmatrix}0&1\\1&0\end{psmallmatrix}$ and $\sigma_z=\begin{psmallmatrix}1&0\\0&-1\end{psmallmatrix}$ denote the usual Pauli matrices. Here, the constant $\lambda\in\IR$ is the coupling of spin and field, whereas $\mu\in\IR$ is the strength of the external magnetic field. Using standard estimates and the Kato--Rellich theorem, it follows that $H(\lambda,\mu)$ is a selfadjoint operator with domain $\cD(\Id\otimes\dG(\omega))$ for all values of $\lambda,\mu\in\IR$, see \cite[Lemma 2.2]{HaslerHinrichsSiebert.2021a} for details. Note that the assumption $v\in \cD(\omega^{-1/2})$ is crucial in this perturbative argument, since it implies that $\ph(v)$ is infinitesimally bounded w.r.t. $\dG(\omega)$.

For our main theorem, we will need the following assumption on the dispersion relation.
\begin{hyp}\label{mainhyp}\label{mainhyp:omega}
	There exists a nowhere dense set $M$ such that, for all $\varrho\in \cC^\infty_M$ and $p\in\IR^d$ with $|p|$ sufficiently small,
	\begin{equation}\label{eq:w-limit}
		\Delta_\varrho(p) = \esssup_{k\in\IR^d} \varrho(k)\frac{\omega(k+p)-\omega(k)}{\omega(k)} < \infty  \qquad \mbox{and}\qquad \Delta_\varrho(p) \xrightarrow{|p|\to0} 0.
	\end{equation}
%	\begin{enumhyp}
%		\item\label{mainhyp:omega} There exists a nowhere dense set $M$ such that for all $\varrho\in \cC^\infty_M$ and $p\in\IR^d$ with $|p|$ sufficiently small
%		\begin{equation}\label{eq:w-limit}
%		\Delta_\varrho(p) = \esssup_{k\in\IR^d} \varrho(k)\frac{\omega(k+p)-\omega(k)}{\omega(k)} < \infty  \qquad \mbox{and}\qquad \Delta_\varrho(p) \xrightarrow{|p|\to0} 0.
%		\end{equation}
%		\item\label{mainhyp:diff} For $|p|$ sufficiently small, $\tau_pv \in  \cD(\omega^{-1/2})$ and
%		\begin{equation}\label{eq:v-limit}
%				\lim_{|p|\to 0}\|\omega^{-1/2}(\tau_p v - v)\| = 0.
%		\end{equation}
%	\end{enumhyp}
\end{hyp}
\begin{rem}
	We note that \cref{eq:w-limit} is especially satisfied if $\omega$ is differentiable with bounded derivative and strictly positive outside of any nowhere dense set, by the mean value theorem.
\end{rem}
Our main \lcnamecref{mainthm} summarizes all positive results for existence of ground states in the spin boson model with external magnetic field.
\begin{thm}\label{mainthm}
	Assume \cref{mainhyp} holds.
	\begin{enumthm}
		\item\label{mainthm:infregular} {\em (Infrared-Regular Case)} If $\omega^{-1}v\in L^2 (\IR^d)$ and $\|\omega^{-1}(\tau_pv-v)\|_2\xrightarrow{|p|\to 0}0$, then $H(\lambda,\mu)$ has a ground state for all values of $\lambda,\mu\in\IR$.
		%\item\label{mainthm:absence} {\em (Infrared-Critical Case I - Symmetry)} If $v\in \cD(\omega^{-1/2})\setminus \cD(\omega^{-1})$, then $H(\lambda,\mu)$ can only have a ground state if $\mu=0$.
		\item\label{mainthm:existence} {\em (Infrared-Critical Case)} Assume $\omega(k)=\omega(-k)$ and $v(k)=\bar{v(-k)}$ for almost all $k\in\IR^d$. Further, assume that $\|\omega^{-1/2}(\tau_pv-v)\|_2 \xrightarrow{|p|\to 0} 0$. If $|\lambda|<\|\omega^{-1/2}v\|_2^{-1}/\sqrt 5$, then $H(\lambda,0)$ has a unique ground state.
	\end{enumthm}
\end{thm}
\begin{rem}
	The existence of ground states for the infrared-regular case is a standard result, see for example \cite{Spohn.1989,BachFroehlichSigal.1998b,Gerard.2000}. 
	Since our compactness argument in \cref{sec:compact} provides a simple proof of this case under very general assumptions, we treat it here nevertheless.
	%	One proof for the existence of ground states for the case $\omega^{-1/2}v\in L^2(\IR^d)$ can be found in \cite{Gerard.2000}. In this spirit, it is worth comparing our assumptions to the ones therein. First, \subcref{mainhyp:omega} follows easily from the mean value theorem if $\omega$ is differentiable with bounded derivative and positive outside of any nowhere dense set, which is the assumption \cite[(H1)]{Gerard.2000}. Further, our assumption \subcref{mainhyp:diff} can be deduced from \cite[(I1)]{Gerard.2000}. Hence, the proof presented here is a slightly different 
\end{rem}
\begin{rem}\label{rem:absence}
	In the infrared-critical case $v\in\cD(\omega^{-1/2})\setminus \cD(\omega^{-1})$, a ground state $\psi$ can only exist if $\braket{\psi,(\sigma_x\otimes\Id)\psi}=0$, see for example \cite{AraiHirokawaHiroshima.1999}. Hence, the only situation in which a ground state might exist due to cancellations of divergences is the case $\mu=0$, which we treat in \subcref{mainthm:existence} of above \lcnamecref{mainthm}. For an extended discussions of the absence of ground states due to infrared divergences, we refer to \cite{Spohn.1998,LorincziMinlosSpohn.2002,Hinrichs.2022}. We emphasize that the expectation $\braket{\psi,(\sigma_x\otimes\Id)\psi}$ vanishes in the case $\mu=0$, due to the underlying spin-flip symmetry in this case, i.e., that $H(\lambda,0)$ commutes with $\sigma_x\otimes(-1)^{\dG(1)}$. This plays a crucial role in the proofs of \cref{thm:resder,cor:dercor} below.
\end{rem}
%\begin{rem}
%	\Cref{mainhyp:simp} ensures that \cref{def:SB} defines a selfadjoint lower-semibounded operator on $\IC^2 \otimes \FS$.
%\end{rem}
\begin{ex}[Physical Example] Let us discuss the typical physical example in $d=3$ dimensions. The massless dispersion relation is given by $\omega(k)=|k|$, whereas the form factor is $v(k)=\kappa(k)|k|^{-\alpha}$ with $\kappa:\IR^d\to[0,1]$ being a suitable cutoff function ensuring $v\in L^2(\IR^d)$, such as the characteristic function of a ball around zero or $\kappa(k)=e^{-ck^2}$ for some $c>0$. Given this choice, the model is infrared-regular for $\alpha\in(0,\frac 12)$. Our result on the infrared-critical case holds for any choice of $\alpha\in[\frac 12,1)$, where the physically most interesting case is given by $\alpha=\frac 12$. We note that for $\alpha\ge 1$, the assumption $v\in\cD(\omega^{-1/2})$ is not satisfied and hence the interaction is not bounded w.r.t. the free operator anymore.
\end{ex}
%\begin{rem}
%	In the case $\omega^{-1}v\in L^2(\IR^d)$, then the existence of a ground state for $H(\lambda,0)$ for arbitrary values of $\lambda\in\IR$ is well-known \cite{.}. By similar arguments, this result can be extended to any value of $\mu\in\IR$.
%\end{rem}
%\begin{rem}
%	Conjecture
%\end{rem}
%\begin{rem}
%	\Cref{mainthm:existence} is essentially the statement of \cite[Theorem XX]{HaslerHinrichsSiebert.2021c}, where the results from \cite{HaslerHinrichsSiebert.2021a,HaslerHinrichsSiebert.2021b,HaslerHinrichsSiebert.2021c} are combined. However, due to a simple improvement of the compactness argument in \cite{HaslerHinrichsSiebert.2021a} using the Fr\'echet--Kolmogorov theorem instead of the Rellich criterion, we can drop the assumption $\omega^{-1/2}v\in L^{2+\eps}(\IR^d)$ for some $\eps>0$ in this note.
%\end{rem}

For the remainder of this note, we discuss the proof of \cref{mainthm}. %, with an emphasis on \subcref{mainthm:existence}.% The proof of \cref{mainthm:absence} is obtained using the result in \cite{AraiHirokawaHiroshima.1999}.
In \cref{sec:compact}, we present the compactness argument which is essential for both parts of the statement and based on \cite{HaslerHinrichsSiebert.2021a,HiroshimaMatte.2019}. The main result therein is stated in \cref{thm:resolventbound}. In \cref{sec:resolventbound}, we then review the results from \cite{HaslerHinrichsSiebert.2021b,HaslerHinrichsSiebert.2021c} and show how to verify the assumption of \cref{thm:resolventbound} in the infrared-critical case for coupling constants smaller than the critical value.
In \cref{sec:outlook}, we summarize the proof of \cref{mainthm} and conjecture the behavior of the spin boson model at large coupling.

Since many details of our proofs are deferred to the articles \cite{HaslerHinrichsSiebert.2021a,HaslerHinrichsSiebert.2021b,HaslerHinrichsSiebert.2021c}, we point the reader to \cite{Hinrichs.2022} for a detailed and extensive proof of \cref{mainthm:existence}.

\section{The Compactness Argument}\label{sec:compact}

In this section, we present the compactness argument which is essential to our proof for existence of ground states. It can essentially be found in \cite{HaslerHinrichsSiebert.2021a} and  is based on \cite{GriesemerLiebLoss.2001}. However, similar to \cite{Matte.2016,HiroshimaMatte.2019}, we replace the Rellich criterion by the Fr\'echet--Kolmogorov theorem, which allows for a simpler proof and more general choices of the dispersion relation $\omega$.

The main idea is to approximate the massless model by the massive one, i.e., by introducing an artificial mass of the bosonic field. Then it is a well-known result that a ground state exists, due to the presence of a spectral gap.
\begin{thm}\label{thm:massive}
	If $\essinf_{k\in\IR^d} \omega(k) >0$, then $H(\lambda,\mu)$ has a unique ground state for all $\lambda,\mu\in\IR$.
\end{thm}
\begin{proof}
	Proofs for this statement in the case $\mu=0$ can, for example, be found in \cite{AraiHirokawa.1995,DamMoller.2018a}. An explicit proof including the external magnetic field is given in \cite[Appendix D]{HaslerHinrichsSiebert.2021c}.
\end{proof}
The value $m_\omega= \essinf_{k\in\IR^d} \omega(k)$ can be interpreted as a boson mass. We are interested in taking $m_\omega\to 0$. To ensure strong convergence of the ground states in this limit, we need to prove that they form a relatively compact set. Therefore, we will employ the following compactness criterion.
%
%
%Now, sequence $(\omega_n)$ corresponding ground states $(\psi_n)$.
%
%In this section, we essentially review the result from \cite{HaslerHinrichsSiebert.2021a}. However, we simplify the compactness argument by applying the Fr\'echet--Kolmogorov theorem instead of the Rellich criterion.
A similar (or in fact stronger) statement can be found in \cite[Proposition 3.8]{HiroshimaMatte.2019}. Nevertheless, for the convenience of the reader, we give the nice proof here, which is based on the Fr\'echet--Kolmogorov theorem.

For the use therein and later on, we define the pointwise annihilator $a_k$. Given $\psi^{(n+1)}\in\FS^{(n+1)}$ for some $n\in\IN_0$, the map $k\mapsto \psi^{(n+1)}(k,\cdots)$ is an element of $L^2(\IR^d;\FS^{(n)})$, by the Fubini--Tonelli theorem. Hence, the prescription
\begin{equation}\label{eq:pointann}
	a_k \psi^{(n+1)} (k_1,\ldots,k_n)  = \sqrt{n+1}\psi^{(n+1)}(k,k_1,\ldots,k_n)
\end{equation}
yields a well-defined element $a_k\psi^{(n+1)}\in \FS^{(n)}$ for almost all $k\in\IR^d$.
The definition directly implies
\begin{equation}\label{eq:annnorm}
	\|\psi^{(n)}\|^2_{\FS^{(n)}} = \frac{1}{n}\int \|a_k\psi^{(n)}\|^2_{\FS^{(n-1)}}\d k
\end{equation}
 We remark that, for $\psi\in\FS$, the vector $a_k\psi = (a_k\psi^{(n+1)})_{n\in\IN_0}$ is not necessarily an element of $\FS$.
% However, the inequality
%\begin{equation}
%	\|\psi^{(n+1)}\|_{\FS^{(n+1)}}\le \frac1{n+1}\|a_k\psi\| \qquad\mbox{for all}\ \psi\in\FS
%\end{equation}
%holds, where the right hand side can be infinity.

The following statement can be found in the standard literature and will be employed in the following proofs: If $A:\IR^d\to[0,\infty)$ is measurable and $\psi\in\cD(\dG(A)^{1/2})$, then $a_k\psi\in \FS$ for almost every $k\in\IR^d$ and
\begin{equation}\label{eq:operatordecomp}
	\|\dG(A)^{1/2}\psi\|^2 = \int A(k)\|a_k\psi\|^2 \d k.
\end{equation}
\begin{thm}\label{thm:compactness}
	Let $\cI\subset \FS$ be bounded and let $M\subset \IR^d$ be nowhere dense. Assume that there exists an $f\in L^2 (\IR^d)$ such that
	\begin{equation}
		\|a_{k}\psi\| \le |f(k)| \qquad  \mbox{for all}\ \psi\in\cI\ \mbox{and}\ \mbox{almost every}\ k\in \IR^d. \label{ass:prop.compact-number}
	\end{equation}
	Further, assume that for all $\varrho\in\cC^\infty_M$ there exists $(g_p^{(\varrho)})_{p\in \IR^d}\subset L^2(\IR^d)$ with $\|g^{(\varrho)}_p\|_2\xrightarrow{|p|\to 0} 0$ such that%, for all $\psi \in M$ and $\varrho\in \cC^\infty_M$,
	\begin{equation}
		 \|\varrho(k+p)a_{k+p}\psi-\varrho(k)a_k\psi\|\le |g^{(\varrho)}_p(k)| \qquad\mbox{for all}\ \psi\in\cI\ \mbox{and}\ \mbox{almost every}\ k\in \IR^d.\label{ass:prop.shift}%  \mbox{for almost all}\ k,p\in \IR^d.
	\end{equation}
	Then $\cI$ is relatively compact.
\end{thm}
\begin{proof}	
	Throughout this proof, given a measurable $B:\IR^d\to[0,1]$, we define the contraction operator
	\begin{equation}
		\Gamma(B) = \bigoplus_{n=0}^\infty \Gamma^{(n)}(B) \qquad\mbox{with}\ \Gamma^{(0)}(B)=1,\ \Gamma^{(n)}(B)(k_1,\ldots,k_n)=B(k_1)\cdots B(k_n)\ \mbox{for}\ n\in\IN
	\end{equation}
	acting on $\FS$.
	
	We split the proof into two steps.
	\begin{step}\label{step:compact.1}
		We fix $n\in\IN$ and $\varrho \in \cC^\infty_M$ and
		prove that $\{(\Gamma(\varrho)\psi)^{(n)}:\psi\in \cI\}$ is a relatively compact subset of $\FS^{(n)}$. %, where
		%\[ \Gamma(\varrho) = \bigoplus_{n=0}^\infty \varrho^{\otimes n}. \]
		
		By the Fr\'echet--Kolmogorov theorem, this is equivalent to showing that
		\begin{align}
			&\sup_{\psi\in\cI}\|\chr{\{|\cdot|>R\}}(\Gamma(\varrho)\psi)^{(n)}\|_{\FS^{(n)}} \xrightarrow{R\to \infty}0\label{eq:FK-local}\\
			&\sup_{\psi\in\cI}\|(\tau_p-\Id)(\Gamma(\varrho)\psi)^{(n)}\|_{\FS^{(n)}}\xrightarrow{\IR^{d\cdot n}\ni|p|\to 0} 0.\label{eq:FK-shift}
		\end{align}
		where $\chr{\{|\cdot|>R\}}$ is treated as multiplication operator on $\FS^{(n)}=L^2_{\sfs\sfy\sfm}(\IR^{d\cdot n})$.
		
		For arbitrary $\psi\in\cI$, using \cref{eq:annnorm,ass:prop.compact-number} as well as the permutation symmetry of $\psi^{(n)}$, we find
		\[ \|\chr{\{|\cdot|>R\}}(\Gamma(\varrho)\psi)^{(n)}\|_{\FS^{(n)}}^2 \le \int_{|k|>R/n} \|a_k\Gamma(\varrho)\psi\|^2\d k  \le \|\chr{|\cdot|>R/n}f\|_2^2.  \]
		Taking $R\to 0$, this implies \cref{eq:FK-local}.
		
		By a similar argument, for $p=(p_1,\ldots,p_n)\in\IR^{d\cdot n}$ and $\psi\in\cI$, we find
		\[ \|(\tau_p-\Id)(\Gamma(\varrho)\psi)^{(n)}\|_{\FS^{(n)}}^2 \le \sum_{\ell=1}^{n} \int \|\varrho(k+p_\ell)a_{k+p_\ell}\psi - \varrho(k)a_k\psi\|^2 \d k \le \sum_{\ell=1}^{n} \|g^{(\varrho)}_{p_\ell}\|^2,    \]
		using \cref{ass:prop.shift}.
		Therefore, $\|g^{(\varrho)}_p\|_2\xrightarrow{|p|\to 0}0$ proves \cref{eq:FK-shift}.
		
		This finishes the first step.
	\end{step}
	\begin{step}
		We can now prove the statement.
		
		To that end, let $(\psi_\ell)_{\ell\in\IN}$ be a sequence in $\cI$ and fix some  $\eps>0$. We choose $N_\eps\in\IN$, $M\subset M_\eps\subset \IR^d$ and $\varrho_\eps \in \cC^\infty_M$ such  that $\|f\|_2^2/(N_\eps+1)<\eps$, $\|\chr{M_\eps} f\|^2<\eps$ and $\varrho_\eps = 1$ on $M_\eps^\sfc$, which is possible since $M$ is nowhere dense. W.l.o.g. (otherwise restrict to a subsequence), by the boundedness of $\cI$ and \cref{step:compact.1}, we can assume $(\psi_\ell)$ is weakly convergent to some $\Psi\in \FS$ and that $\slim\limits_{\ell\to\infty}(\Gamma(\varrho_\eps)\psi_\ell)^{(n)} = (\Gamma(\varrho_\eps)\Psi)^{(n)}$ for all $n=0,\ldots,N_\eps$.
		
		Let $P_\eps$ denote the orthogonal projection onto the subspace $\bigoplus\limits_{n=0}^{N_\eps}\FS^{(n)}$ of $\FS$. This implies
		\begin{equation}\label{eq:stronglimit}
			\begin{aligned}
			\|\Psi\|^2  &\ge \|P_\eps \Gamma( \varrho_\eps ) \Psi\|^2 = \lim_{\ell\to \infty} \|P_\eps \Gamma(\varrho_\eps)\psi_\ell\|^2 = \lim_{\ell\to \infty} \braket{\psi_\ell,P_\eps\Gamma(\varrho_\eps^2)\psi_\ell}  \\& \ge  \limsup_{\ell\to \infty}\left(\|\psi_\ell\|^2-\|(1-P_\eps)\psi_\ell\|^2 \right) - \sup_{\phi\in\cI}\braket{\phi,P_\eps(1-\Gamma(\varrho_\eps^2))\phi}.
			\end{aligned}
		\end{equation}
		Now, by \cref{ass:prop.compact-number,eq:annnorm}, we have
		\[ \|(1-P_\eps)\psi_\ell\|^2 = \sum_{n=N_\eps+1}^\infty\frac{1}{n}\int \|a_k\psi_\ell^{(n)}\|^2\d k \le \frac{1}{N_\eps+1}\int \|a_k\psi_\ell\|^2 \d k \le \frac{\|f\|^2_2}{N_\eps+1} < \eps. \]
		Further, looking at the $n$-particle subspace $\FS^{(n)}$, we see
		\[ \Gamma^{(n)}(\varrho_\eps^2)(k_1,\ldots,k_n) = \prod_{j=1}^{n}\varrho_\eps^2(k_j) \ \begin{cases} = 0 & \mbox{if}\ \exists j\in\{2,\ldots,n\}:k_j\in(\supp \varrho_\eps)^\sfc, \\ = 1 &\mbox{if}\ k_1,\ldots,k_n\in M_\eps^\sfc,\\ \in[0,1]& \mbox{else}.\end{cases}  \]
		Hence, $1-\Gamma(\varrho_\eps^2) \le \dG(\chr{M_\eps})$. Using \cref{eq:operatordecomp,ass:prop.compact-number}, this implies
		\[  \braket{\phi,P_\eps(1-\Gamma(\varrho_\eps^2))\phi} \le \int_{M_\eps} \|a_k\phi\|^2 \d k \le \|\chr{M_\eps}f\|^2 < \eps \qquad\mbox{for all}\ \phi\in\cI.\]
		Inserting these observations into \cref{eq:stronglimit}, we find $\|\Psi\|^2 \ge \limsup_{\ell\to \infty} \|\psi_\ell\|^2 - 2 \eps$. Taking $\eps\to 0$, we have $\|\Psi\|\ge \limsup_{\ell \to \infty}\|\psi_\ell\|$. Further, the weak lower-semicontinuity of norms and $\Psi = \wlim_{\ell\to\infty}\psi_\ell$ imply $\|\Psi\|\le \liminf_{\ell\to\infty}\|\psi_\ell\|$. Hence, $\|\Psi\| = \lim_{\ell\to\infty}\|\psi_\ell\|$, which proves $(\psi_\ell)$ in fact converges to $\Psi$ strongly.
		
		This finishes the proof.
	\end{step}
\end{proof}
%\begin{prop}
%	Let $(\omega_n)$ be a sequence of functions such that:
%	\begin{enumerate}
%		\item $(\omega_n,v)$ satisfy the assumptions in \cref{mainhyp} with $\omega= \omega_n$.
%		\item $\omega_n$ 
%	\end{enumerate}
%\end{prop}
%
%\begin{proof}[\textbf{Proof of \cref{mainthm:infregular}}]
%\end{proof}
We now want to approximate $H(\lambda,\mu)$ by a sequence of spin boson Hamiltonians with massive bosons. Hence, we fix a sequence $(\omega_n)_{n\in\IN}$ satisfying
\begin{enumhyp}
	\item $\omega_n:\IR^d\to[0,\infty)$ is measurable, positive almost everywhere and $v\in\cD(\omega_n^{-1/2})$ for all $n\in\IN$,
	\item $\omega_n(k+p)-\omega_n(k)\le \omega(k+p)-\omega(k)$ for all $k,p\in\IR^d$, $n\in\IN$,
	\item $(\omega_n)$ is (a.e. pointwise) monotonically decreasing and converges to $\omega$ uniformly,
	\item $\essinf_{k\in\IR^d} \omega_n(k)>0$ for all $n\in\IN$.
\end{enumhyp}
We remark that possible choices for this sequence are $\omega_n = \sqrt{\omega^2 + m_n^2}$ or $\omega_n = \omega + m_n$ for some strictly decreasing zero-sequence $(m_n)\subset (0,\infty)$. The parameter $m_n>0$ can be interpreted as the artificial boson mass described above.

By \cref{thm:massive}, the operators $H_n(\lambda,\mu)$ obtained from \cref{def:SB} by replacing $\omega$ by $\omega_n$ have a unique ground state for each $n\in\IN$ and all values $\lambda,\mu\in\IR$. From now on, let $\pln$ be a normalized ground state of $H_n(\lambda,\mu)$ and denote $E_n(\lambda,\mu)=\inf\sigma(H_n(\lambda,\mu))$.
%Now,
Using \cref{thm:compactness}, the convergence of $(\pln)$ to the ground state of $H(\lambda,\mu)$ can be deduced if a resolvent bound is satisfied.
\begin{thm}\label{thm:resolventbound} Fix $\lambda,\mu\in\IR$. Further,
	assume \cref{mainhyp:omega} holds. If there exists a measurable $h:\IR^d\to[0,\infty)$ such that $v,\tau_p v\in \cD(h)$ for $p\in\IR^d$ small enough, $\|h(\tau_pv-v)\|_2\xrightarrow{|p|\to 0} 0$ and
	\begin{equation}\label{eq:resbound}
		\|(H_n(\lambda,\mu)-E_n(\lambda,\mu)+\omega_n(k))^{-1}(\sigma_x\otimes \Id) \psi\| \le h(k) \qquad\mbox{for almost all}\ k\in \IR^d,
	\end{equation}
	then $H(\lambda,\mu)$ has a unique ground state.
	%If $\|R_n(\lambda)\sigma_x\psi_n\|\le \sfC \omega_n(k)^{-1/2}$ for all $n\in\IN$, then $H(\lambda,0)$ has a ground state.
\end{thm}
\begin{rem}
	In our proof of \cref{mainthm}, we will (up to a positive constant) choose $h\in\{\omega^{-1},\omega^{-1/2}\}$.
\end{rem}
\begin{proof}
	Throughout this proof, we fix $\lambda,\mu\in\IR$ and write
	\[ R_n(k) =  (H_n(\lambda,\mu)-E_n(\lambda,\mu)+\omega_n(k))^{-1}.\]
	Further, we drop tensor products with the identity.
	
	The connection between the considered resolvents and the assumptions of \cref{thm:compactness} is given by the pull-through formula
	\begin{equation}\label{eq:pt}
		a_k\psi_n^{(\lambda,\mu)} = -v(k)R_n(k)\sigma_x \psi_n^{(\lambda,\mu)} \qquad\mbox{for almost all}\ k\in\IR^d,
	\end{equation}
	where $a_k$ acts on $\pln$ component-wise in the sense of \cref{def:HS}
	For proofs, see for example \cite{BachFroehlichSigal.1998b,Gerard.2000} or \cite[Lemma 6.13]{Hinrichs.2022}.
	
	We now want to apply \cref{thm:compactness} to the bounded set $\cI=\{\psi_n^{(\lambda,\mu)}:n\in\IN\}$, where the nowhere dense set $M$ is chosen as in \cref{mainhyp}.
	By the assumption, \cref{ass:prop.compact-number} follows from \cref{eq:pt} with the choice $f=hv$. Further, by the resolvent identity, we find
	\begin{align*}
		\varrho(k+p)a_{k+p}\pln - \varrho(k)a_k\pln  = &\, \varrho v(k)R_n(k)\sigma_x\pln - \varrho v(k+p)R_n(k+p)\sigma_x\pln\\
		= &\, (\varrho v (k) - \varrho v(k+p))R_n(k)\sigma_x \pln \\ &\, + \varrho v(k+p)R_n(k+p)(\omega_n(k+p)-\omega_n(k))R_n(k)\sigma_x\pln.
	\end{align*}
	Hence, \cref{mainhyp}, the standard bound
	\begin{equation}\label{eq:standard}
		\|R_n(k)\|^{-1}\le \omega_n(k)^{-1} \le \omega(k)^{-1}
	\end{equation}
	and the assumptions on $(\omega_n)$ imply
	\begin{align*}
		\|\varrho(k+p)a_{k+p}\pln &- \varrho(k)a_k\pln\|\\& \le h(k)(v(k)-v(k+p)) + \varrho(k+p)\frac{\omega_n(k)-\omega_n(k+p)}{\omega_n(k+p)} h(k)v(k+p)\\
		& \le h(k)(v(k)-v(k+p)) + \varrho(k+p)\frac{\omega(k)-\omega(k+p)}{\omega(k+p)} h(k)v(k+p)
	\end{align*}
	Therefore, \cref{ass:prop.shift} is satisfied with the choice $g^{(\varrho)}_{p} = h(\tau_pv-v) + \Delta_\varrho(p)hv$, where $\Delta_\varrho(p)$ is defined as in \cref{eq:w-limit}. This proves $\cI$ is relatively compact.

	W.l.o.g., we can now assume that the sequence $(\pln)$ is strongly convergent to some normalized $\psi_0\in \FS$. It remains to prove that $\psi_0$ is a ground state of $H(\lambda,\mu)$. To that end, we first observe that, by \cref{eq:operatordecomp,eq:pt,eq:standard}, we have
	\[ \|\dG(\omega)^{1/2}\pln\|^2 = \int \omega(k)\|a_k\pln\|^2 \le \|v\|_2^2.  \]
	This implies that $\psi_0\in\cD(\dG(\omega)^{1/2})=\cD((H(\lambda,\mu)-E(\lambda,\mu))^{1/2})$. Further, by the lower-semicontinuity of closed quadratic forms and using both the monotonicity and the uniform convergence of $(\omega_n)$, we find 
	\begin{align*}
		\|(H(\lambda,\mu)-E(\lambda,\mu))^{1/2}\psi_0\|^2 &= \braket{\psi_0,(H(\lambda,\mu)-E(\lambda,\mu))\psi_0} \\&\le \liminf_{n\to\infty}\braket{\pln,(H(\lambda,\mu)-E(\lambda,\mu))\pln}\\ & \le \liminf_{n\to\infty}\braket{\pln,(H_n(\lambda,\mu)-E(\lambda,\mu))\pln}\\
		&\le \liminf_{n\to\infty}(E_n(\lambda,\mu)-E(\lambda,\mu)) = 0.
	\end{align*}
	Hence, $(H(\lambda,\mu)-E(\lambda,\mu))^{1/2}\psi_0=0\in \cD((H(\lambda,\mu)-E(\lambda,\mu))^{1/2})$ and therefore $\psi_0\in\cD(H(\lambda,\mu))$ with $H(\lambda,\mu)\psi_0=E(\lambda,\mu)\psi_0$.
	
	The uniqueness of the ground state can be inferred from positivity arguments, see \cite[Proposition 3.6]{HaslerHinrichsSiebert.2021c} for details.
	This finishes the proof.
\end{proof}

\section{Proof of Resolvent Bound}\label{sec:resolventbound}

In this \lcnamecref{sec:resolventbound}, we show how to prove the resolvent bound \cref{eq:resbound} with $h=\omega^{-1/2}$ if $\mu=0$ and $|\lambda|$ is smaller than a critical value. This is essential for the proof of existence of ground states in the infrared critical case, cf. \cref{mainthm:existence}. Proofs for the statements presented in this \lcnamecref{sec:resolventbound} can be found in \cite{HaslerHinrichsSiebert.2021b,HaslerHinrichsSiebert.2021c}. Many of our results only hold for the case of massive bosons, i.e., if
\[ m_\omega := \essinf_{k\in\IR^d} \omega(k)>0. \]
We will emphasize this fact, if this is the case. %Throughout this \lcnamecref{sec:resolventbound}, we use $(\omega_n)$, $H_n$ and $E_n$ as defined in the previous section.
Throughout this section, we write
\[E(\lambda,\mu) = \inf\sigma(H(\lambda,\mu)).\]
The first observation we make is that the desired resolvent bound follows, if a bound on the second derivative of the ground state energy w.r.t. the external magnetic field is satisfied. This statement does not carry over to the case $\mu\ne 0$, since the proof significantly builds on the spin-flip symmetry in the case $\mu=0$. To ensure the differentiability of the ground state energy, we need to assume the Hamiltonian has a spectral gap, which is the case if $m_\omega>0$. The proof uses simple techniques from perturbation theory.

\begin{thm}[{\cite[Lemmas 4.3, 4.4]{HaslerHinrichsSiebert.2021a}}]\label{thm:resder}
 	Let $\lambda\in\IR$ and assume $m_\omega>0$. Further, assume $\psi$ is a normalized ground state of $H(\lambda,0)$, which exists by \cref{thm:massive}.
 	Then $E(\lambda,\mu)$ is analytic in $\mu$ and we have \[ \|(H(\lambda,0)-E(\lambda,0)+\omega(k))^{-1}\sigma_x\psi\| \le (-\partial_\mu^2E(\lambda,\mu)\big|_{\mu=0})^{1/2}  \omega^{-1/2}(k) \qquad\mbox{for all}\ k\in\IR^d. \]
\end{thm}
The remainder of this section, we prove that the derivatives $\partial_\mu^2 E(\lambda,\mu)$ are bounded uniformly in $m_\omega$. To that end, we first use the well-known connection between the operator  semigroup and the ground state energy of an operator, also referred to as Bloch's formula, to rewrite these derivatives as expectation values w.r.t. the ground state of the free operator $H(0,0)$. We refrain from giving a proof here, but emphasize that the generalization of Bloch's formula to derivatives of the ground state energy is again only possible due to the spectral gap of $H(\lambda,\mu)$. In the statement, let
\[\Od = \begin{pmatrix}0\\1\end{pmatrix} \otimes (1,0,0,\ldots)\in\IC^2\otimes \FS  \]
denote the ground state of the free operator $H(0,0)$.
\begin{thm}[{\cite[Lemma 2.9, Theorem 2.10]{HaslerHinrichsSiebert.2021c}}]\label{thm:semigroup}Fix $\lambda,\mu\in\IR$.\\
	Then
	\[ E(\lambda,\mu) = \lim_{T\to \infty} -\frac 1T\partial_\mu \ln\Braket{\Od,e^{-TH(\lambda,\mu)}\Od}.  \]
	Further, if $m_\omega>0$, we have
	\[ \partial_\mu^\ell E(\lambda,\mu) = \lim_{T\to \infty} -\frac 1T\partial_\mu^\ell \ln\Braket{\Od,e^{-TH(\lambda,\mu)}\Od} \qquad\mbox{for all}\ \ell\in\IN. \]
\end{thm}
To prove that the expectation on the right hand side is bounded, we now exploit the connection to a continuous one-dimensional Ising model. To make this statement more precise, let $(X_t)_{t>0}$ be a continuous-time Markov process taking values in $\{\pm 1\}$, satisfying $\PP[X_0=1]=\PP[X_0=-1]=\frac 12$ and having Poisson-distributed jump times with parameter 1, i.e.,
\[ \PP[X_t=x|X_s=y] = \frac 12 (1+\delta_{x,y}e^{-2|t-s|}-\delta_{x,-y}e^{-2|t-s|}) \qquad\mbox{for}\ t,s>0,\ x,y\in\{\pm 1\}, \]
where $\delta_{\cdot,\cdot}$ denotes the usual Kronecker delta.
The connection to the operator semigroup is now given by a Feynman--Kac--Nelson (FKN) type formula. Whereas we prove the full FKN formula in \cite{HaslerHinrichsSiebert.2021c}, we here only present the version obtained when treating the expectation w.r.t. $\Od$ and integrating out the field degrees of freedom. We emphasize that this statement also holds in the massless case $m_\omega=0$.
\begin{thm}[{\cite[Corollary 2.6]{HaslerHinrichsSiebert.2021c}}]\label{thm:FKN} Assume $\omega(k)=\omega(-k)$ and $v(k)=\bar{v(-k)}$ holds for almost all $k\in\IR^d$.
	Further, we define
	\[ W(t)= \frac 14 \int_{\IR^d}|v(k)|^2e^{-|t|\omega(k)}\d k.\]
	 Then, for all $\lambda,\mu\in\IR$ and $T>0$, we have
	\[ e^{-T}\Braket{\Od,e^{-TH(\lambda,\mu)}\Od} = \EE\left[\exp\left(\lambda^2 \int_0^T\int_0^T W(t-s)X_tX_s\d t \d s - \mu \int_0^T X_t \d t\right)\right].  \]
\end{thm}
\begin{rem}
	The integrals on the right hand side are Riemann integrals. If $X$ is realized as a random variable on a probability space $\cP$, then the function $t\mapsto X_t(p)$ has only finitely many discontinuities in the interval $[0,T]$ for almost every $p\in\cP$. Combined with the continuity of $W$, this implies the existence of the integrals. Further, the continuity also implies that the expression on the right hand side is uniformly bounded in $p\in\cP$. Hence, the expectation value exists and is finite, by the dominated convergence theorem.
\end{rem}
\begin{rem}\label{rem:Ising}
	The right hand side can be interpreted as the partition function of an Ising model over $\IR$ with interaction function $W$ and external magnetic field $\mu$. This connection between spin boson and Ising models has for example been used in \cite{EmeryLuther.1974,FannesNachtergaele.1988,Spohn.1989,Abdessalam.2011,HirokawaHiroshimaLorinczi.2014}.
\end{rem}

Combining \cref{thm:FKN,thm:semigroup}, we find that the derivatives of the ground state energy can be calculated as correlation functions of the continuous Ising model described in \cref{rem:Ising}. To simplify notation, given a continuous function $I:\IR\to[0,\infty)$, let us define the partition function
\[ Z^{(I)}_T(\lambda,\mu) =  \EE\left[\exp\left(\lambda^2 \int_0^T\int_0^T I(t-s)X_tX_s\d t \d s - \mu \int_0^T X_t \d t\right)\right]. \]
Further, given a random variable $Y$ defined on the same probability space as $X$, we define its expectation value
\[\llangle Y\rrangle^{(I)}_{T,\lambda,\mu} = \frac{1}{Z^{(I)}_T(\lambda,\mu)}\EE\left[Y\exp\left(\lambda^2 \int_0^T\int_0^T I(t-s)X_tX_s\d t \d s - \mu \int_0^T X_t \d t\right)\right].\]
\begin{cor}\label{cor:dercor}
	Let $W$ be as defined in \cref{thm:FKN}. Further, we assume $m_\omega>0$ and $\omega(k)=\omega(-k)$ as well as $v(k)=\bar{v(-k)}$ for almost all $k\in\IR^d$. Then
	\[ \partial_\mu^2 E(\lambda,\mu)|_{\mu=0} = -\lim_{T\to\infty}\frac 1T\left\llangle \left(\int_0^T X_t\d t\right)^2\right\rrangle^{(W)}_{T,\lambda,0} \qquad\mbox{for all}\ \lambda\in\IR.  \]
\end{cor}
\begin{proof}
	From the definition, it is easy to observe that $Z^{(W)}_T(\lambda,\mu)$ is infinitely often differentiable w.r.t. $\mu$ and the first two derivatives are
	\begin{align*}
		&\partial_\mu Z^{(W)}_T(\lambda,\mu) = -\EE\left[\int_0^T X_t\d t \exp\left(\lambda^2 \int_0^T\int_0^T I(t-s)X_tX_s\d t \d s - \mu \int_0^T X_t \d t\right) \right]\\
		&\partial_\mu Z^{(W)}_T(\lambda,\mu) = \EE\left[\left(\int_0^T X_t\d t\right)^2 \exp\left(\lambda^2 \int_0^T\int_0^T I(t-s)X_tX_s\d t \d s - \mu \int_0^T X_t \d t\right) \right]
	\end{align*}
	Since, by definition of $X$, both $X$ and $-X$ are equivalent stochastic processes, this implies \[\partial_\mu Z^{(W)}_T(\lambda,\mu)|_{\mu=0} = 0.\]
	Now inserting the second derivative into \cref{thm:semigroup} and applying \cref{thm:FKN} on the right hand side, we find
	\begin{align*}
		\partial_\mu^2 E(\lambda,\mu)
		& = -\lim_{T\to \infty}\frac 1T \partial_\mu^2( \ln Z^{(W)}_T(\lambda,\mu) + 1) \\
		& = - \lim_{T\to \infty}\frac 1T \frac{\partial_\mu^2Z^{(W)}_T(\lambda,\mu)}{Z^{(W)}_T(\lambda,\mu)}. 
	\end{align*}
	Combining above observations proves the statement.
\end{proof}
Let us finish this \lcnamecref{sec:resolventbound}, with the last ingredient of our proof.
Treating the continuous-time Ising model as a continuum limit of a discrete Ising model, we prove the following upper bound for correlation functions in \cite{HaslerHinrichsSiebert.2021b}.
\begin{thm}[{\cite[Theorem 1.2]{HaslerHinrichsSiebert.2021b},\cite[Theorem 4.21]{Hinrichs.2022}}]\label{thm:corbound}
	For all $\eps\in(0,\frac 15)$, there exists $C_\eps>0$ such that for all continuous and even $I\in L^1(\IR)$ with $I\ge 0$ and $\|I\|_1<\eps$, we have
	\[ \limsup_{T\to\infty}\frac 1T \left\llangle\left(\int_0^T X_t\d t\right)^2\right\rrangle^{(I)}_{T,\lambda,0} \le C_\eps. \]
\end{thm}
Let us now combine the results presented in this \lcnamecref{sec:resolventbound}.
\begin{cor}\label{cor:resbound}
	Assume $\omega(k)=\omega(-k)$, $v(k)=\bar{v(-k)}$ for almost all $k\in\IR^d$. Further, let $\lambda\in\IR$ with $|\lambda|<\|\omega^{-1/2}v\|^{-1}/\sqrt 5$. Then, there exists a $C>0$ such that the following holds:\\
	Let $\wt\omega$ be any dispersion relation satisfying $m_{\wt\omega}>0$, $\wt\omega(\cdot)=\wt\omega(-\cdot)$ as well as $\wt\omega\ge \omega$ almost everywhere. Further, let
	\[ \wt H(\lambda,\mu) = \sigma_z\otimes\Id + \Id\otimes\dG(\wt\omega) + \sigma_x\otimes (\lambda\ph(v)+\mu\Id) \qquad\mbox{and}\qquad \wt E(\lambda,\mu) = \inf\sigma(\wt H(\lambda,\mu)) \]
	 %$\wt H(\lambda,\mu)$ and $\wt E(\lambda,\mu)$ denote the spin boson Hamiltonian and its ground state energy defined in \cref{def:SB} and \cref{def:E}, respectively, but with $\omega$ replaced by $\wt \omega$. 
	 and let $\wt\psi$ be a normalized ground state of $\Ht(\lambda,0)$, which exists due to \cref{thm:massive}. Then
	\[ \|(\wt H(\lambda,0)-\wt E(\lambda,0)+\wt\omega(k))^{-1}\sigma_x\wt\psi\|\le C\wt\omega(k)^{-1/2} \qquad\mbox{for all}\ k\in\IR^d. \]
\end{cor}
\begin{rem}
	We emphasize that we do not assume $m_\omega>0$. Hence, the uppper bound especially holds for any approximating sequence $(\omega_n)$ of $\omega$ as described in the previous section and is uniform in $n$.
\end{rem}
\begin{proof}
	Since $\wt\omega$ satisfies the assumptions of \cref{cor:dercor}, we have
	\[ \partial_\mu^2 \wt E(\lambda,\mu)\big|_{\mu=0} = -\lim_{T\to\infty}\frac 1T \left\llangle\left(\int_0^TX_t\d t\right)^2 \right\rrangle^{(\wt W)}_{T,\mu,0} \qquad\mbox{with}\ \wt W(t) = \frac 14\lambda ^2\int |v(k)|^2 e^{-|t|\wt\omega(k)}\d k.\]
	Further, we easily calculate $\|\wt W\|_1 = \lambda^2 \|\wt\omega^{-1/2}v\|_2^2 \le \lambda^2 \|\omega^{-1/2}v\|_2^2 $. If $|\lambda|<\|\omega^{-1/2}v\|^{-1}/\sqrt 5$, then there is an $\eps\in(0,\frac 15)$ such that $\|\wt W\|_1 < \eps$. Hence, we can apply \cref{thm:corbound} and there exists $C>0$ such that
	\[ 0\ge \partial_\mu^2 \wt E(\lambda,\mu)\big|_{\mu=0} \ge -C. \]
	Inserting this into \cref{thm:resder} finishes the proof.
\end{proof}

\section{Conclusion \& Conjecture}\label{sec:outlook}

We now show how to combine the arguments presented in the previous two \lcnamecrefs{sec:resolventbound} to the
\begin{proof}[\textbf{Proof of \cref{mainthm}}]
	The statement \subcref{mainthm:infregular} directly follows from \cref{thm:resolventbound} and the standard resolvent bound \cref{eq:standard}, setting $h=\omega^{-1}$. Further, \subcref{mainthm:existence} follows from \cref{thm:resolventbound}, by setting $h=C\omega^{-1/2}$, where $C$ is the constant from \cref{cor:resbound}. In this case, the resolvent bound is satisfied since we can set $\wt\omega=\omega_n$ in \cref{cor:resbound} for any $n\in\IN$.
\end{proof}

In \cref{mainthm,rem:absence}, we have intentionally left one case open. Explicitly, it is an open problem to rigorously treat the spin boson model without magnetic field at large coupling. However, some heuristic intuition comes from the connection to a long-range one-dimensional Ising model presented in \cref{thm:FKN,rem:Ising}. 

Considering the physical infrared-critical situation in $d=3$ dimensions with $\omega(k)=|k|$ and $v(k)\sim |k|^{-1/2}$ for small $|k|$, the interaction function $W(t)$ defined in \cref{thm:FKN} decays as $1/t^2$ as $t\to\infty$. In the discrete Ising model, it is well-known that long-range Ising models with interaction decaying this way exhibit a phase transition \cite{AizenmanChayesChayesNewman.1980,ImbrieNewman.1988}. In this sense, we do not expect that the $\eps$ in \cref{thm:corbound} can be taken arbitrarily large. Therefore, it is reasonable to assume that the resolvent bound required in \cref{thm:resolventbound} and proven in \cref{cor:resbound} might cease to hold for $|\lambda|$ becoming large.

Hence, we conclude this note with the following
\begin{conj}
	Assume $v\in\cD(\omega^{-1/2})\setminus\cD(\omega^{-1})$. Then there exists $\lambda_\sfc>0$ such that $H(\lambda,0)$ has no ground state for $|\lambda|>\lambda_\sfc$.
\end{conj}
	It should be mentioned that a similar result for KMS states was proven in \cite{Spohn.1989}. Therefore, the conjectured statement is also expected from this point of view.

%\section{Concluding Remarks and Outlook}

\bibliographystyle{halpha-abbrv}
\bibliography{../../Literature/00lit}

\end{document}

%% file: header.tex
\usepackage{amsmath,amssymb,mathtools,bm}
\usepackage[utf8]{inputenc}
\usepackage{csquotes}
\usepackage[english]{babel}
\usepackage[T1]{fontenc}
\usepackage{amstext}
\usepackage{amsthm}
\usepackage{bbm}
\usepackage{enumerate}
\usepackage{hyperref}
\usepackage[nameinlink,capitalise]{cleveref}
\usepackage{braket}
\usepackage{dsfont}
\usepackage{color}
 \usepackage{tikz}
\usepackage{pgfplots}

\usepackage{authblk}

\usepackage[verbose=true,letterpaper]{geometry}
\AtBeginDocument{
	\newgeometry{
		textheight=9.5in,
		textwidth=7in,
		top=1in,
		headheight=14.5pt,
		headsep=10pt,
		footskip=20pt
	}
}

\usepackage{fancyhdr}
\fancyheadoffset{0pt}
\usepackage{titlesec}
\titleformat{\section}{\bfseries\scshape\Large}{\thesection}{1em}{}{}
\titleformat*{\subsection}{\scshape\bfseries\large}

\renewenvironment{thebibliography}[1]{
	\begin{oldthebibliography}{#1}
		\footnotesize
		\setlength{\itemsep}{0em}
		\setlength{\parskip}{0em}
	}
	{
	\end{oldthebibliography}
}

\numberwithin{equation}{section}
\newtheorem{thm}{Theorem}[section]
\newtheorem{lem}[thm]{Lemma}
\newtheorem{prop}[thm]{Proposition}
\newtheorem{cor}[thm]{Corollary}

\theoremstyle{definition}
\newtheorem{defn}[thm]{Definition}
\newtheorem{hyp}{Hypothesis}
\renewcommand*{\thehyp}{\Alph{hyp}}

\newtheorem*{conj}{Conjecture}

\theoremstyle{remark}
\newtheorem{rem}[thm]{Remark}
\newtheorem{ex}[thm]{Example}

\crefname{hyp}{Hypothesis}{Hypotheses}
\Crefname{hyp}{Hypothesis}{Hypotheses}
\crefname{lem}{Lemma}{Lemmas}
\Crefname{lem}{Lemma}{Lemmas}
\crefname{thm}{Theorem}{Theorems}
\Crefname{thm}{Theorem}{Theorems}
\crefname{prop}{Proposition}{Propositions}
\Crefname{prop}{Proposition}{Propositions}
\crefname{enumi}{}{}
\Crefname{enumi}{}{}
\creflabelformat{enumi}{#2(#1)#3}
\crefname{equation}{}{}
\Crefname{equation}{}{}
\crefname{rem}{Remark}{Remarks}
\Crefname{rem}{Remark}{Remarks}

\makeatletter
\renewcommand{\@upn}{} % to use the same font for the number as for the head
\makeatother

\usepackage{etoolbox}
\makeatletter
\patchcmd{\endthm}{\@endpefalse}{}{}{}
\patchcmd{\endcor}{\@endpefalse}{}{}{}
\patchcmd{\endlem}{\@endpefalse}{}{}{}
\patchcmd{\endprop}{\@endpefalse}{}{}{}
\patchcmd{\endproof}{\@endpefalse}{}{}{}
\makeatother

\usepackage[inline]{enumitem}
\newlist{enumthm}{enumerate}{1} % set up a dedicated enumeration env.
\setlist[enumthm]{label=\upshape(\roman*),ref=\thethm~(\roman*)}  %labels are upshape, references as environment
\crefalias{enumthmi}{thm} % alias 'enumthmi' counter to 'thm'
\newlist{enumcor}{enumerate}{1}
\setlist[enumcor]{label=\upshape(\roman*),ref=\thecor~(\roman*)}
\crefalias{enumcori}{cor}
\newlist{enumlem}{enumerate}{1}
\setlist[enumlem]{label=\upshape(\roman*),ref=\thelem~(\roman*)}
\crefalias{enumlemi}{lem}
\newlist{enumprop}{enumerate}{1}
\setlist[enumprop]{label=\upshape(\roman*),ref=\theprop~(\roman*)}
\crefalias{enumpropi}{prop}
\newlist{enumhyp}{enumerate}{1}
\setlist[enumhyp]{label=\upshape(\roman*),ref=\thehyp~(\roman*)}
\crefalias{enumhypi}{hyp}
\newlist{enumproof}{enumerate*}{1}
\setlist[enumproof]{label=\upshape(\roman*)}
\newlist{enumdef}{enumerate}{1}
\setlist[enumdef]{label=\upshape(\roman*),ref=\thedefn~(\roman*)}
\crefalias{enumdefi}{defn}

\makeatletter
\newcounter{subcreftmpcnt} %
\newcommand\romansubformat[1]{(\roman{#1})} %adapt ....
\def\subcref{\@ifstar\@@subcref\@subcref}
\newcommand\@subcref[2][\romansubformat]{%
	\ifcsname r@#2@cref\endcsname
	\cref@getcounter {#2}{\mylabel}%
	\setcounter{subcreftmpcnt}{\mylabel}%
	\hyperref[#2]{\romansubformat{subcreftmpcnt}}%
	\else ?? \fi}   
\newcommand\@@subcref[2][\romansubformat]{%
	\ifcsname r@#2@cref\endcsname
	\cref@getcounter {#2}{\mylabel}%
	\setcounter{subcreftmpcnt}{\mylabel}%
	\romansubformat{subcreftmpcnt}%
	\else ?? \fi}   
\makeatother

\makeatletter
\DeclareRobustCommand{\crefnosort}[1]{%
	\begingroup\@cref@sortfalse\cref{#1}\endgroup
}
\makeatother

\def\endstepsymbol{$\lozenge$}
\def\endclaimsymbol{$\lozenge$}
\newcounter{proofstep}
\AtBeginEnvironment{proof}{\setcounter{proofstep}{0}}
\newenvironment{step}{\refstepcounter{proofstep}\\[.2em]\noindent
	{\em Step \arabic{proofstep}.}}{\endstepsymbol}
\crefname{proofstep}{Step}{Steps}
\Crefname{proofstep}{Step}{Steps}
\newcounter{proofclaim}
\AtBeginEnvironment{proof}{\setcounter{proofclaim}{0}}

\crefname{proofclaim}{Claim}{Claims}
\Crefname{proofclaim}{Claim}{Claims}

\input{alphabets}

\newcommand{\IN}{\BN}\newcommand{\IR}{\BR}\newcommand{\IC}{\BC}

\newcommand{\PP}{\DSP}\newcommand{\EE}{\DSE}

\newcommand{\HS}{\cH}

\newcommand{\eps}{\varepsilon}\newcommand{\ph}{\varphi}

\newcommand{\Id}{\dsone} \renewcommand{\d}{\sfd}

\newcommand{\supp}{\operatorname{supp}}\DeclareMathOperator*{\esssup}{ess\,sup}\DeclareMathOperator*{\essinf}{ess\,inf}

\DeclareMathOperator*{\slim}{s-lim}\DeclareMathOperator*{\wlim}{w-lim}

\newcommand{\wt}[1]{\widetilde{#1}}\renewcommand{\bar}[1]{\overline{#1}}

\DeclareFontFamily{U}{mathx}{\hyphenchar\font45}
\DeclareFontShape{U}{mathx}{m}{n}{
	<5> <6> <7> <8> <9> <10>
	<10.95> <12> <14.4> <17.28> <20.74> <24.88>
	mathx10
}{}
\DeclareSymbolFont{mathx}{U}{mathx}{m}{n}
\DeclareFontSubstitution{U}{mathx}{m}{n}
\DeclareMathAccent{\widecheck}{0}{mathx}{"71}
\DeclareMathAccent{\wideparen}{0}{mathx}{"75}

\DeclareFontFamily{OMX}{MnSymbolE}{}
\DeclareFontShape{OMX}{MnSymbolE}{m}{n}{
	<-6>  MnSymbolE5
	<6-7>  MnSymbolE6
	<7-8>  MnSymbolE7
	<8-9>  MnSymbolE8
	<9-10> MnSymbolE9
	<10-12> MnSymbolE10
	<12->   MnSymbolE12}{}
\DeclareSymbolFont{mnlargesymbols}{OMX}{MnSymbolE}{m}{n}
\SetSymbolFont{mnlargesymbols}{bold}{OMX}{MnSymbolE}{b}{n}
\DeclareMathDelimiter{\llangle}{\mathopen}{mnlargesymbols}{'164}{mnlargesymbols}{'164}
\DeclareMathDelimiter{\rrangle}{\mathclose}{mnlargesymbols}{'171}{mnlargesymbols}{'171}
\DeclareMathDelimiter{\lsem}{\mathopen}{mnlargesymbols}{'102}{mnlargesymbols}{'102}
\DeclareMathDelimiter{\rsem}{\mathclose}{mnlargesymbols}{'107}{mnlargesymbols}{'107}
\DeclareMathDelimiter{\langlebar}{\mathopen}{mnlargesymbols}{'152}{mnlargesymbols}{'152}
\DeclareMathDelimiter{\ranglebar}{\mathclose}{mnlargesymbols}{'157}{mnlargesymbols}{'157}
\DeclareMathDelimiter{\lWavy}{\mathopen}{mnlargesymbols}{'137}{mnlargesymbols}{'137}
\DeclareMathDelimiter{\rWavy}{\mathopen}{mnlargesymbols}{'137}{mnlargesymbols}{'137}

\newcommand{\FGamma}{\Gamma}
\newcommand{\FS}{\cF}\newcommand{\dG}{\sfd\FGamma}

%% file: alphabets.tex
\newcommand{\cC}{{\mathcal C}}
\newcommand{\cD}{{\mathcal D}}\newcommand{\cF}{{\mathcal F}}
\newcommand{\cH}{{\mathcal H}}\newcommand{\cI}{{\mathcal I}}

\newcommand{\cP}{{\mathcal P}}

\newcommand{\BC}{{\mathbb C}}

\newcommand{\BN}{{\mathbb N}}
\newcommand{\BR}{{\mathbb R}}

\usepackage{dsfont} %can be called with option sans for dsletters without serifs

\newcommand{\DSE}{{\mathds E}}

\newcommand{\DSP}{{\mathds P}}

\newcommand{\dsone}{{\mathds 1}}

\usepackage{mathrsfs}

\newcommand{\sfc}{{\mathsf c}}
\newcommand{\sfd}{{\mathsf d}}

\newcommand{\sfm}{{\mathsf m}}

\newcommand{\sfs}{{\mathsf s}}

\newcommand{\sfy}{{\mathsf y}}